\RequirePackage{etex}

\documentclass[12pt,psamsfonts]{amsart}
\usepackage{amsfonts,pstricks}
\usepackage{amsmath, amsthm, amssymb,  mathrsfs, epsfig}
\usepackage{dcpic, pictex, pinlabel}
\usepackage{hyperref}
\usepackage{graphicx}
\usepackage{color}
\usepackage{ifpdf}
\usepackage{tikz}
\usepackage{xypic}
\usepackage{extarrows}
\usepackage{multirow}

\begin{document}

\newcommand{\Z}{\mathbb{Z}}
\newcommand{\M}{\mathcal{M}}
\newcommand{\B}{\mathcal{B}}
\newcommand{\Sym}{\mathcal{S}}
\newcommand{\A}{\mathcal{A}}
\newcommand{\SUM}{\textrm{SUM}}
\newcommand{\CNOT}{\textrm{CNOT}}
\newcommand{\Flip}{\textrm{Flip}}
\newcommand{\SU}{SU}
\newcommand{\SO}{SO}
\newcommand{\C}{\mathbb{C}}
\newcommand{\Cat}{\mathcal{C}}

\theoremstyle{definition}
\newtheorem{axiom}{Axiom}
\newtheorem{thm}{Theorem}
\newtheorem{Conjecture}{Conjecture}
\newtheorem{lem}{Lemma}
\newtheorem{example}{Example}
\newtheorem{cor}{Corollary}
\newtheorem{prop}{Proposition}
\newtheorem{rem}{Remark}
\newtheorem{definition}{Definition}
\newtheorem{measurement}{Measurement}
\newtheorem{ancilla}{Ancilla}

\numberwithin{equation}{section} \makeatletter
\renewenvironment{proof}[1][\proofname]{\par
    \pushQED{\qed}%
    \normalfont \topsep6\p@\@plus6\p@ \labelsep1em\relax
    \trivlist
    \item[\hskip\labelsep\indent
        \bfseries #1]\ignorespaces
}{%
    \popQED\endtrivlist\@endpefalse
} \makeatother
\renewcommand{\proofname}{Proof}

\title{Universal Quantum Computation with Metaplectic Anyons}
\author{Shawn X. Cui$^{1}$ and Zhenghan Wang$^{1,2}$}

\address{$^1$Department of Mathematics\\University of California\\Santa Barbara, CA 93106}
\email{xingshan@math.ucsb.edu, zhenghwa@math.ucsb.edu}
\address{$^2$Microsoft Research, Station Q\\ University of California\\ Santa Barbara, CA 93106}
\email{zhenghwa@microsoft.com}

\thanks{The authors are partially supported by NSF DMS 1108736.}
\date{Revised Feb. 26th 2015}


\maketitle

\begin{abstract}
We show that braidings of the metaplectic anyons $X_\epsilon$ in $\SO(3)_2=\SU(2)_4$ with their total charge equal to the metaplectic mode $Y$ supplemented with projective measurements of the total charge of two metaplectic anyons are universal for quantum computation.  We conjecture that similar universal anyonic computing models can be constructed for all metaplectic anyon systems $\SO(p)_2$ for any odd prime $p\geq 5$.  In order to prove universality, we find new conceptually appealing universal gate sets for qutrits and qupits.
\end{abstract}

\section{Introduction}

Anyons are modeled mathematically by simple objects of unitary modular categories (UMCs).  An important invariant of an anyon type $x$ is its quantum dimension $d_x$---the ground state degeneracy $V_{n,x}$ of $n$ type $x$ anyons in the disk $D^2$ (with an appropriate total charge) is asymptotically $d_x^n$.  When $d_x=1$, an anyon of type $x$ is abelian.  Otherwise, $d_x>1$ and such anyons are non-abelian.  A non-abelian anyon of type $x$ naturally leads to a representation $\rho_{n,x}$ of the $n$-stand braid group $B_n$: $B_n\rightarrow U(V_{n,x})$ for each $n\geq 1$.  The property $F$ conjecture is that the images $\rho_{n,x}(B_n)$ in $U(V_{n,x})$ are all finite subgroups if and only if $d_x^2 \in \Z$ \cite{naidu11}.  When $d_x^2 \in \Z$, then anyons of type $x$ are called {\it weakly integral}.  Interesting weakly integral anyons include those in metaplectic UMCs \cite{hastings2013metaplectic}, which are known to have Property $F$ \cite{rowell14}.

Anyons can be used for quantum information processing.  Ideally, we would like to have a non-abelian anyon such as the Fibonacci anyon whose braidings alone are universal for quantum computation \cite{FLW02}.  But more realistic anyons seem to be weakly integral.  If the Property $F$ conjecture holds, then weakly integral anyons cannot be universal for quantum computation by braidings alone.  Therefore, it is interesting to investigate what extra resources are required for universal quantum computation.  In \cite{cui2014universal}, we analyze the simplest integral non-abelian UMC $D(S_3)$.  In this paper, we focus on anyons in the metaplectic UMCs.  We separate weakly integral anyons into two classes \cite{footnote1}: $P$-anyons and $\#P$-anyons.  $P$-anyons are those whose associated link invariants can be computed classically in polynomial time, while the associated link invariants of $\#P$-anyons are $\#P$ hard to compute.  In particular all abelian anyons are $P$-anyons.  Abelian anyons are only good for topological quantum memory because the resulting braid group representations only lead to phases\cite{footnote2}.  While the Ising anyon leads to many topologically protected quantum gates, all can be simulated classically efficiently because they are Clifford gates \cite{nielsen2010quantum}.  Moreover, we believe the projective measurements of the total charge of any number of Ising anyons can also be simulated classically efficiently.  The Ising anyon and the metaplectic anyon $X_\epsilon$ of quantum dimension $=\sqrt{p}$ in $\SO(p)_2$ are all $P$-anyons.  Surprisingly, the metaplectic modes $Y_i$ of quantum dimension $=2$ in the metaplectic UMCs are $\#P$-anyons \cite{hastings2013metaplectic}.  This $\#P$-hardness makes us believe that if the metaplecitc modes $Y_i$'s are used in the computation, we might gain extra computational power.  Indeed, we will show that braidings of the metaplectic anyons $X_\epsilon$ in $\SO(3)_2=\SU(2)_4$ with their total charge equal to the metaplectic mode $Y$ supplemented with projective measurements of the total charge of two metaplectic anyons are universal for quantum computation.  We conjecture that similar universal computing models can be constructed for all metaplectic anyon systems $\SO(p)_2$ for any odd prime $p\geq 5$.  In order to prove universality, we find new conceptually appealing universal gate sets for qutrits and qupits.

Our interest for $\SU(2)_4$ comes from its potential physical relevance.  There are many possible routes to realize our universal quantum computational model: fractional quantum Hall liquids at $\nu=8/3$ \cite{RR99}, bilayer fractional quantum Hall liquids at $\nu=2/3$  \cite{BW10}, metaplectic anyons \cite{HNW13}, and parafermion zero modes \cite{CAS}.  Evidence for the realization of $SU(2)_4$ in fractional quantum Hall liquids at $\nu=8/3$ is found numerically \cite{peterson15}.  The experimental challenge is to find a realization of the metaplectic mode $Y$ of dimension $=2$.

The paper is organized as follows. In Section \ref{univeral gate set}, we give two new universal gate sets, one for qutrit and one for qupit.  In Section \ref{background}, we present a universal anyonic model with $\SU(2)_4$ and use theorems in Section \ref{univeral gate set} to prove its universality. We also propose a similar model with $\SO(5)_2$ and provide some partial results. Appendices \ref{6jSO(3)},\ref{6jSO(5)} contain the $6j$-symbols and $R$-symbols for $\SO(3)_2$ and $\SO(5)_2,$ which are the data we need to compute the braid group representations and construct braiding quantum gates. Appendix \ref{braid matrix} shows how to compute the braid matrices for $1$-qudit models.

\section{Universal gate sets for qutrits and qupits} \label{univeral gate set}

Throughout this paper, $d\geq 2$ is an integer and $\omega_d = e^{\frac{2\pi i}{d}}$ is the $d$-th root of unity.  We will set $\omega=\omega_3$, and use $p$ to denote an odd prime $p\geq 5$.

Let $\C^d$ be the qudit with the standard basis $\{|j\rangle|j=0,1,...,d-1\}$.  For $p$ an odd prime $p\geq 5$, we will refer to a qudit as a {\it qupit}.  It is not hard to believe that qubits and qutrits behave differently from qupits.  Our universal gate sets below show some differences already.

A standard universal gate set for the qubit quantum circuit model consists of the Hadamard gate $H$, the controlled-NOT gate $\CNOT$, and the $\frac{\pi}{8}$-gate $T$ \cite{boykin1999universal} \cite{nielsen2010quantum}.  There are natural generalizations of the Hadamard and $\CNOT$ gates to qudits.  The $T$ gate is a $4$-th root of the Pauli $\sigma_z$ matrix.  If we propose generalizations of the Pauli $\sigma_z$ to qudits, how many roots do we need to take for obtaining a universal gate set?  For our generalizations of the Pauli matrix, the answer is simply $2$ for qutrits and none for qupits.

The generalized Hadamard gate for qudits is the generalized Hadamard gate $H_d$:

$$H_d \;|j\rangle = \frac{1}{\sqrt{d}}\sum\limits_{i=0}^{d-1} \omega_d^{ij}|i\rangle, \, j = 0,1,\cdots, d-1.$$

A natural generalization of the $\CNOT$ gate is the following $\SUM$ gate:

$$ \SUM_d \; |i,j\rangle = |i,i+j (\textrm{mod} \, d)\rangle, \, i,j = 0,1,\cdots, d-1.$$

The $T$-gate is the $4$-th root of the Pauli $\sigma_z$ matrix.  The  $\sigma_z$ gate can be generalized to the $Q[i]$ gates for qudits:

$$Q[i]_d \; |j\rangle = \omega_d^{\delta_{ij}}|j\rangle, \, i,j = 0,1,\cdots, d-1.$$

Related to the $Q[i]$ gates are the $P[i]$ gates:

$$P[i]_d \; |j\rangle = (-\omega_d^2)^{\delta_{ij}}|j\rangle, \, i,j = 0,1,\cdots, d-1.$$

Some other gates that will be used throughout this paper are:

The generalized $X$ gate, $X_d \; |i\rangle = |i+1 (\textrm{mod} \, d)\rangle,$

The generalized $Z$ gate, $Z_d \; |i\rangle = \omega_d^i |i\rangle,$

The generalized controlled-$Z$ gate, $\bigwedge(Z)_d \; |i,j\rangle = \omega_d^{ij}|i,j\rangle,$

Sign-flip gate, $\Flip[i]_d \; |j\rangle = (-1)^{\delta_{ij}} |j\rangle, \, i = 0,1,\cdots, d-1.$

When $d=3$, the $P[i]$ gate is a square root of the $Q[i]$ gate.  In general when $p$ is an odd prime, $Q[i]_p$ is always a power of $P[i]_p.$
When no confusion arises, we will drop the subscripts $d$ or $p$ from the notation.

We will prove below that for $d=3$ the gate set consisting of the generalized Hadamard gate $H_3$, the $\SUM$ gate $\SUM_3$, and any one of the $P[i]_3, i=0,1,2$ gates is universal for the qutrit quantum circuit model, while for qupits the generalized Hadamard gate $H_p$, the $\SUM$ gate $\SUM_p$, and the $Q[i]_p, i=1,2,..., p-1$ gates suffice.  Our universal qutrit gate set is new.  The universal qupit gate set is distilled from the universal gate sets in \cite{aharonov9906129fault}, though our universal gate set is not explicitly given there and our proof of universality is new.

\subsection{Universal qutrit gate sets} \label{universal qutrit}

\begin{thm} \label{qutrit thm}
The following gate set is universal for the qutrit quantum circuit model:

1). The generalized Hadamard gate $H_3$

2). The $\SUM$ gate $\SUM_3$

3). Any gate from the set $\{P[0]_3, \, P[1]_3, \,P[2]_3\}$.
\end{thm}

\begin{rem}
The universal set above has a strong analogy with the standard qubit universal set $\{\textrm{CNOT}, H, T = \pi/8$-$\textrm{gate}\}$ in that $\{\SUM_3, H_3, P[2]_3^2\}$ generate the qutrit Clifford group while $\{\textrm{CNOT}, H, T^2\}$ generate the qubit Clifford group. In this sense, our universal qutrit gate set above is a natural generalization of the standard universal qubit set.
\end{rem}

To prove the theorem, we need the following lemmas.

\begin{lem} \label{non-commuting SU2} \cite{kitaev1997quantum}
Let $U_1, U_2$ be two non-commuting matrices in $\SU(2).$ If they are both of infinite order, then the subgroup generated by $U_1, U_2$ is dense in $\SU(2).$
\end{lem}

\begin{lem} \label{stablizer universal} \cite{kitaev1997quantum}
Let $V$ be any finite dimensional Hilbert space. Let $H \subset \SU(V)$ be the stabilizer of some non-zero vector $|\psi\rangle \in V$ and $U \in \SU(V)$ be any operator which does not preserve the space $span\{|\psi\rangle\},$ then the set of operators $\{ H \bigcup U^{-1}HU \}$ generate a dense subgroup of $\SU(V)$.
\end{lem}

\begin{definition} \cite{brylinski2002universal}
1). A vector $|\psi\rangle \in \mathbb{C}^d \otimes \mathbb{C}^d$ is called decomposable if $|\psi\rangle = |\psi_1\rangle \otimes |\psi_2\rangle$ for some $|\psi_1\rangle, |\psi_2\rangle \in \mathbb{C}^d.$

2). A quantum gate $U \in \textrm{U}(\mathbb{C}^d \otimes \mathbb{C}^d)$ is primitive if it maps decomposable states to decomposable states. Otherwise, $U$ is called imprimitive.

\end{definition}

\begin{lem} \label{SUM}
The gate $\SUM_d$  is imprimitive.
\begin{proof}
Consider the decomposable state $\frac{1}{\sqrt{d}}\sum\limits_{i=0}^{d-1} |i\rangle \otimes |0\rangle.$ It is mapped, by $\SUM_d$, to $\sum\limits_{i=0}^{d-1}\frac{1}{\sqrt{d}}|i\rangle \otimes |i\rangle,$ which is not a decomposable state.
\end{proof}
\end{lem}

Set $W[i] = H_3 P[i]_3 H_3^{-1} P[i]_3^{-1}, \, Z[i] = H_3 P[i]_3^{-1} H_3^{-1} P[i]_3, \, i = 0,1,2.$

\begin{lem} \label{1-qutrit universal}
The generalized Hadamard $H_3$ and any gate from $\\ \{P[0]_3, \, P[1]_3, \,P[2]_3\}$ generate a dense subgroup of $\SU(3)$.
\begin{proof}
Direct calculations show that $W[i]$ and $Z[i]$ both have eigenvalues $\{\frac{2 \pm i \sqrt{5}}{3}, 1\}.$ Moreover, $W[i]$ and $Z[i]$ share an eigenvector with eigenvalue $1$, which is the following vector $E_i$, respectively, for $i = 0, 1, 2$,

$E_1 = -|1\rangle + |2\rangle, \quad E_2 = -\omega|0\rangle + |2\rangle, \quad E_3 = -\omega|0\rangle + |1\rangle. $

Clearly, $\frac{2 \pm i \sqrt{5}}{3}$ are the roots of the irreducible polynomial $3x^2 - 4x+3, $ which is not a cyclotomic polynomial. Thus $\frac{2 \pm i \sqrt{5}}{3}$ are not roots of unity. Restricted to $E_i^{\bot},$ the two dimensional orthogonal complement of $E_i$, $W[i]$ and $Z[i]$ are of  infinite order. It is straightforward to check that $W[i]$ and $Z[i]$ do not commute. By Lemma \ref{non-commuting SU2}, $W[i]$ and $Z[i]$ generate a dense subgroup of $\SU(E_i^{\bot}).$

Since $H_3$ does not preserve $span\{E_i\},$ it follows from Lemma \ref{stablizer universal} that $\SU(E_i^{\bot}) \bigcup H_3^{-1} \SU(E_i^{\bot}) H_3$ generate a dense subgroup of $\SU(3).$ Therefore, $\{H_3, P[i]_3\}$ generate a dense subgroup of $\SU(3)$.

\end{proof}
\end{lem}

\begin{proof}[Proof of Theorem \ref{qutrit thm}]
By Theorem $1.3$ in \cite{brylinski2002universal}, the collection of $1$-qudit gates with any imprimitive $2$-qudit gate form a universal gate set for $d \geq 3$. By Lemma \ref{SUM}, $\SUM_3$ is an imprimitive 2-qutrit gate. By Lemma \ref{1-qutrit universal}, $H_3$ and any gate from $\{P[0]_3, \, P[1]_3, \,P[2]_3\}$ generate a dense subgroup of the group of all $1$-qutrit gates. Thus, the gates from our theorem form a universal gate set.
\end{proof}

To state the next theorem, we introduce a qutrit coherent projective measurement.

\begin{measurement} \label{measure 0 12}
The projection of a state in the qutrit space $\C^3$ to $span\{|0\rangle\}$ and its orthogonal complement $span\{|1\rangle, |2\rangle\}$ so that the resulting state, if it is in $span\{|1\rangle, |2\rangle\}$,  is coherent.
\end{measurement}

\begin{thm} \label{qutrit thm2}
The following gate set is universal for the qutrit quantum circuit model.

1). The generalized Hadamard gate $H_3$

2). The $\SUM$ gate $\SUM_3$

3). Any gate from $\{Q[i]_3, i= 0, 1 , 2\}$

4). Any non-trivial $1$-qutrit classical gate not equal to $H_3^2$.

5). Measurement \ref{measure 0 12}

\begin{rem}

\begin{enumerate}
\item In \cite{cui2014universal}, a stronger theorem is proved: the gate set in Theorem $2$ with the gate from $3$) removed is already universal. We proved this stronger theorem by picking a qubit $\C^2$ inside a qutrit $\C^3$ and showing that one can approximate arbitrary unitary $U \in SU(2^n)$. We can then deduce universality for the qutrit circuit by encoding a qutrit with two qubits $\mathbb{C}^2 \otimes \mathbb{C}^2 \subset \mathbb{C}^3 \otimes \mathbb{C}^3$. For instance, we can use $|00\rangle, \, |01\rangle, \, |10\rangle$ to encode $|0\rangle, \, |1\rangle, \, |2\rangle $, respectively. And the basis element $|11\rangle$ is left unused. But it is not known if the reduced qutrit set can be used to approximate arbitrary qutrit gates directly $($$i.e.$ not by encoding a qutrit with two qubits$)$. Neither is it known if the gates from $3$ can be constructed out of the reduced gate set.

\item Comparing this theorem with Corollary \ref{qudit cor} below, we see another difference between qutrit and qupits:  the analogous gates through $1$) to $4$) are already a universal gate set for the qupit quantum circuit model, but not so for the qutrit model.

\item If we restrict the choice of gate from $3$) on $i=1, 2$, then we can drop the gate in $4$) while still keep the universality of the rest. This is because $H_3^2$ is the classical gate which swaps $|1\rangle$ and $|2\rangle$, so with $H_3$ and one of $Q[1]_3$ , $Q[2]_3$, we can obtain the other one. Since $Z_3 = Q[1]_3 Q[2]_3^2$, and $X_3 = H_3^{-1}Z_3H_3$, we can construct the generalized $X$ gate $X_3$, which is a classical gate not equal to $H_3$.

\end{enumerate}
\end{rem}

\begin{proof}
We prove this theorem by showing that we can construct all the gates in Theorem \ref{qutrit thm}. Since $H_3^2$ is a classical gate, the gate from $4)$ together with $H_3^2$ generate all the $1$-qutrit classical gates. It is clear that we only need to construct $P[i]_3$ for some $i.$ Without loss of generality, we assume the gate from $3$) is $Q[2]_3$, since we can permute the basis elements with the classical $1$-qutrit gates. From the identity $P[2]_3 = Q[2]_3 \Flip[2]_3,$ it suffices to construct $\Flip[2]_3$.  The construction of the sign-flip gate was given as an exercise in \cite{KitaevHM} and a detailed proof can be found in Section $2.5$ of \cite{cui2014universal}. For completeness, we also give the proof in Lemma \ref{Flip construction}.
\end{proof}
\end{thm}

\begin{lem}\cite{KitaevHM},\cite{cui2014universal} \label{Flip construction}
The gate $\Flip[2]$ can be constructed probabilistically. Moreover, the probability to construct $\Flip[2]$ approaches to $1$ exponentially fast in the number of gates and measurements given in Theorem \ref{qutrit thm2}.
\begin{proof}
It's not hard to see that with the gates and measurement from Theorem \ref{qutrit thm2}, the following states and operations can be implemented.

$1)$. $\widetilde{|i\rangle} = \frac{1}{\sqrt{3}}\sum\limits_{j=0}^{2}\omega^{ij}|j\rangle = H|i\rangle$, $i = 0, \; 1, \; 2 $.

$2)$. Projection of a $1$-qutrit state to any computational basis vector, preserving the coherence of the orthogonal complement. For example, projection to $span\{|2\rangle\}$ and its complement $span\{|0\rangle, |1\rangle\}$.

$3)$. Measurement of a qutrit in the standard computational basis.

$4)$. Projection to $span\{\widetilde{|1\rangle}, \widetilde{|2\rangle}\}$ and its complement $span\{\widetilde{|0\rangle}\}$

To obtain $\Flip[2]$, we first construct the ancilla $|\psi\rangle = \frac{1}{\sqrt{3}}(|0\rangle - |1\rangle + |2\rangle)$ as follows.

Prepare the state $\widetilde{|1\rangle}\widetilde{|2\rangle},$ and project each qutrit to the space $span\{|0\rangle, |1\rangle\}$ to obtain the state $|\eta\rangle = \frac{1}{2}(|0\rangle + \omega |1\rangle) \otimes (|0\rangle + \omega^2 |1\rangle).$ Apply the $SUM$ gate to $|\eta\rangle$ and then project the first qutrit of the resulting state to the space $span\{\widetilde{|0\rangle}\}$. It's easy to see on the second qutrit we get the state $|\psi\rangle.$

Now for a state $|\phi \rangle = c_0 |0\rangle + c_1|1\rangle + c_2|2\rangle, $ apply the $SUM$ gate to $|\phi\rangle |\psi\rangle$ and then measure the second qutrit in the standard basis. If the outcome is $|0\rangle,$ then the first qutrit is  $c_0 |0\rangle + c_1|1\rangle - c_2|2\rangle. $ If the outcome is $|1\rangle$, then the first qutrit is $-c_0 |0\rangle + c_1|1\rangle + c_2|2\rangle, $ and if the outcome is $|2\rangle,$ then the first qutrit is $c_0 |0\rangle - c_1|1\rangle + c_2|2\rangle $. Moreover, the probability for each case is $\frac{1}{3}$. Therefore, this process changes the sign of some coefficient randomly. Repeat this process until we get the gate $\Flip[2]$. Note that we will also stop repeating the process if we obtain the gate $\Flip[0]*\Flip[1]$, which is the same as $\Flip[2]$ up to a global sign.

Let $p_n$ be the probability that $\Flip[2]$ is obtained $($up to a global sign$)$ with no more than $n$ times of the process. It's not hard to derive a recursive formula for $p_n:$
\begin{equation}
p_n = p_{n-1}+ \frac{1}{3}(1-p_{n-1}) = \frac{1}{3}+\frac{2}{3}p_{n-1}, \quad  p_1 = \frac{1}{3}.
\end{equation}

Therefore, we have $p_n = 1 - (\frac{2}{3})^n$, which approximates to $1$ exponentially fast.

\end{proof}
\end{lem}

\subsection{A universal qupit gate set}

\begin{thm} \label{qudit thm}
The following gate set is universal for the qupit quantum circuit model for $p\geq 5$:

1). The generalized Hadamard gate $H_p$

2). The $\SUM$ gate $\SUM_p$

3). The gates $Q[i]_p, i = 1, \cdots, p-1.$

\begin{rem}
Note that the gate $Q[0]_p$ can be constructed from $Q[i]_p, i = 1, \cdots, p-1,$ since $\prod\limits_{i=0}^{p-1} Q[i]_p = \omega_p Id. $
\end{rem}
\begin{proof}
The proof is analogous to that of Theorem \ref{qutrit thm}. By Lemma \ref{1-qudit universal} below, the gates in $1$) and $3$) generate a a dense subgroup of $\SU(p).$ By Lemma \ref{SUM}, $\SUM_p$ is an imprimitive gate. Again by Theorem $1.3$ in \cite{brylinski2002universal}, this is a universal gate set.
\end{proof}

\end{thm}

\begin{cor} \label{qudit cor}
$H_p, \SUM_p$ and all the $1$-qupit classical gates, together with some $Q[i]_p$, form a universal qupit gate set.
\end{cor}

The rest of this section is devoted to a proof of Lemma $\ref{1-qudit universal}.$

Set $X[i] = H_p Q[i]_p H_p^{-1} Q[i]_p^{-1}, \, Y[i] = H_p Q[i]_p^{-1} H_p^{-1} Q[i]_p,\, i = 0, \cdots, p-1.$ And define
$$
S_i = span\{|i\rangle, \sum\limits_{j \neq i} \omega_p^{ij}|j\rangle\}.
$$

\begin{lem} \label{XY[i]_d} \cite{aharonov9906129fault}
$X[i], Y[i]$ act as the identity on $S_i^{\bot},$ and are of infinite order confined in $S_i$. Moreover, they do not commute.
\begin{rem}
It is worth noting that Lemma \ref{XY[i]_d} does not hold for $d = 3.$ In the case of $d=3$, we need to replace $Q[i]$ by its square root $P[i].$ This is how we defined $W[i], Z[i]$ in Subsection \ref{universal qutrit}.
\end{rem}
\end{lem}

 Given a subspace $A \subset \mathbb{C}^d$, $\SU(A)$ is identified with the subgroup of $\SU(\mathbb{C}^d)$ whose elements are the identity on the orthogonal complement of $A$.

\begin{lem} \label{non-orthogonal} \cite{aharonov9906129fault}
If $A, B$ are two non-orthogonal subspaces of $\mathbb{C}^p$, then $\SU(A) , \SU(B)$ generate a dense subgroup of $\SU(A+B).$
\end{lem}

\begin{lem} \label{1-qudit universal}
The generalized Hadamard gate $H_p$ and the gate set $\{Q[i]_p, i = 0, \cdots, p-1\}$ generate a dense subgroup of $\SU(p).$

\begin{proof}
By Lemma \ref{XY[i]_d} and Lemma \ref{non-commuting SU2}, $X[i]$ and $Y[i]$ generate a dense subgroup of $\SU(S_i).$ It's easy to see that $S_i$ is not orthogonal to $\sum\limits_{j=0}^{i-1} S_j,$ and $\sum\limits_{j=0}^{p-1} S_j = \mathbb{C}^p.$ By induction on $i$ and by Lemma \ref{non-orthogonal}, we obtain a dense subgroup of $\SU(p).$
\end{proof}
\end{lem}

\section{Universal Models from Metaplectic Anyon Systems}\label{background}

We will follow the set-up of anyonic quantum computing models as in \cite{Wang, cui2014universal}, in particular the notations in Section $2$ of \cite{cui2014universal}.  We refer to a particular anyonic model by a pair $(V,b)$, where $V$ is the fusion space that encodes one qudit and $b$ a basis of V designated as the computational basis.  This notation is not complete because we also need to specify the encoding of two qudits, which will be clear from the context.

\subsection{The metaplectic anyon system $\SO(p)_2$}

For a detailed discussion of $\SO(p)_2$, see \cite{hastings2013metaplectic}. The UMC $\SO(p)_2$ for an odd prime $p = 2r+1\geq 5$ has $r+4$ isomorphism classes of simple objects (also called anyon types). We denote the set of simple object representatives by $\{\mathbf{1}, Z, X_{\epsilon}, X_{\epsilon}', Y_j, 1 \leq j \leq r\}$, and their types $\{1,z,\epsilon, \epsilon', y_j, 1 \leq j \leq r\}$.   Their quantum dimensions are $d_{\mathbf{1}} = d_Z = 1, \, d_{X_{\epsilon}} = d_{X_{\epsilon'}} = \sqrt{p}, \, d_{Y_j} = 2.$ We will follow \cite{HNW13} to refer to the anyons $X_{\epsilon}, X_{\epsilon}'$ as the {\it metaplectic anyons}, and the anyons $\{Y_j, 1 \leq j \leq r\}$ as the {\it metaplectic modes}.

The following is a list of some of the fusion rules which are suffcient for deducing all the other fusion rules.
$
\\
1).\, X_{\epsilon} \otimes X_{\epsilon} \cong \mathbf{1} \oplus \oplus_{j=1}^{r} Y_j
\\
2).\, X_{\epsilon} \otimes X_{\epsilon}' \cong Z \oplus \oplus_{j=1}^{r} Y_j
\\
3).\, X_{\epsilon} \otimes Y_j \cong X_{\epsilon} \oplus X_{\epsilon}', \, 1 \leq j \leq r
\\
4).\, X_{\epsilon} \otimes Z \cong X_{\epsilon}'
\\
5).\, Z \otimes Z \cong \mathbf{1}
\\
6).\, Z \otimes Y_j \cong Y_j, \, 1 \leq j \leq r
\\
7).\, Y_j \otimes Y_j \cong \mathbf{1} \oplus Z \oplus Y_{min\{2j, m-2j\}}, \, 1 \leq j \leq r
\\
8).\, Y_i \otimes Y_j \cong Y_{|i-j|} \oplus Y_{min\{i+j, m-i-j\}}, \, 1 \leq i,j \leq r, i \neq j
$
\vspace{.1in}

The UMC $\SO(3)_2$ is the same as $\SU(2)_4$.  There are five anyon types in $\SO(3)_2, $ namely, $1, z, {\epsilon}, {\epsilon}', y.$ Their quantum dimensions are $1,\; 1, \; \sqrt{3},\; \sqrt{3},\; 2.$

\begin{rem}
 The anyon types in $\SU(2)_4$ are usually denoted as $\{0,\;1,\;2,\;3,\;4\}$, which are twice the spin of the corresponding irreps of $\SU(2)$. The correspondence between $\SO(3)_2$ and $\SU(2)_4$ labels are given as follows:

${1} \leftrightarrow 0, \qquad z \leftrightarrow 4, \qquad {\epsilon} \leftrightarrow 1, \qquad {\epsilon}' \leftrightarrow 3, \qquad y \leftrightarrow 2, \qquad$
\end{rem}

We use the fusion tree shown in Figure \ref{1-qutrit pic} to encode a qutrit.

%
%

\begin{figure}
\includegraphics{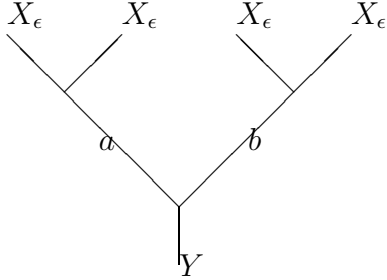}
\caption{1-qutrit model}\label{1-qutrit pic}
\end{figure}

The associated Hilbert space $V^{\epsilon\epsilon\epsilon\epsilon}_y$  is $3$-dimensional with the computational basis $\{-|YY\rangle, |\mathbf{1}Y\rangle, |Y \mathbf{1}\rangle\}.$  We will often use the type labels in fusion spaces.
The computational basis $\{-|YY\rangle, |\mathbf{1}Y\rangle, |Y \mathbf{1}\rangle\}$ is identified with the qutrit basis $\{|0\rangle, |1\rangle, |2\rangle\}, $ respectively. Note that this qutrit basis does not have a $\Z_3$ symmetry. We will denote this $\SU(2)_4$ anyonic computational model as $(V^{\epsilon\epsilon\epsilon\epsilon}_y, \{-|YY\rangle, |\mathbf{1}Y\rangle, |Y \mathbf{1}\rangle\})$ or simply $V^{\epsilon\epsilon\epsilon\epsilon}_y$.

\begin{rem}
The minus sign in front of $|YY\rangle$ is introduced to make the braid representation matrices into a nicer form. This is inessential since the gate $\Flip[0]$ that changes the sign of $|0\rangle$ can be constructed.
\end{rem}

The data needed to analyze the computational power of our model are the $F$-matrices and $R$-symbols. All the $F$-matrices and $R$-symbols for $\SU(2)_4$ are listed in Appendix \ref{6jSO(3)}. Partial data enough for our purpose for $\SO(5)_2$ are listed in Appendix \ref{6jSO(5)}. See Appendix \ref{braid matrix} on how braid matrices for $\sigma_1,\sigma_2, \sigma_3$ in the following subsections are derived.

\subsection{The universal model $V^{\epsilon\epsilon\epsilon\epsilon}_y$ with $\SU(2)_4$} \label{SO(3)}

Under the basis $ \\ \{-|YY\rangle, |\mathbf{1}Y\rangle, |Y \mathbf{1}\rangle\}$, the generators of the braid group $\B_4$ for the representation $V^{\epsilon\epsilon\epsilon\epsilon}_y$ have the following matrices:

$\sigma_1 = \gamma
\begin{pmatrix}
1 &  0  &  0 \\
0 &  \omega &  0 \\
0 &  0  &  1 \\
\end{pmatrix}
\qquad
\sigma_3 =\gamma
\begin{pmatrix}
1 &  0  &  0 \\
0 &  1 &  0 \\
0 &  0  &  \omega \\
\end{pmatrix}
$

$\sigma_2 =\gamma
\begin{pmatrix}
\frac{1}{2} + \frac{\sqrt{3}i}{6} &  -\frac{1}{2} + \frac{\sqrt{3}i}{6}  &  -\frac{1}{2} + \frac{\sqrt{3}i}{6} \\
-\frac{1}{2} + \frac{\sqrt{3}i}{6} & \frac{1}{2} + \frac{\sqrt{3}i}{6}   &  -\frac{1}{2} + \frac{\sqrt{3}i}{6} \\
-\frac{1}{2} + \frac{\sqrt{3}i}{6} &  -\frac{1}{2} + \frac{\sqrt{3}i}{6}  &  \frac{1}{2} + \frac{\sqrt{3}i}{6} \\
\end{pmatrix},
$

where $\gamma = e^{\frac{\pi i}{12}}.$

Note that $\sigma_1,\; \sigma_3$ are just $Q[1]_3,\; Q[2]_3$ defined in Section \ref{univeral gate set}, up to a phase.

The group generated by these matrices is a subgroup of $\SU(3)$ of order $648$ whose center is isomorphic to $\mathbb{Z}_3$. It is isomorphic to the complex reflection group which is the $25$-th item in the classification table of finite complex reflection groups in \cite{shephard1954finite}.  The elements in the center are scalar matrices. And the group modulo the center is isomorphic to the famous Hessian group $\sum(216)$ of order $216$, which is also the $1$-qutrit Clifford group \cite{Jordan} \cite{miller2012theory} \cite{fairbairn1964finite}.  In the following, we will choose many braids whose representation matrices provide us desired gates.  They are obtained by systematically analyzing the representation $V^{\epsilon\epsilon\epsilon\epsilon}_y$ of $\B_4$.

Define $p$ = $\sigma_1\sigma_2\sigma_1$, $q$ = $\sigma_3\sigma_2\sigma_3,$ $H$ = $q^2 p q^2$, then $($ignoring the phase $\gamma$ $),$

$p^2 =-
\begin{pmatrix}
0 &  1  &  0 \\
1 &  0 &  0 \\
0 &  0  &  1 \\
\end{pmatrix}
\qquad
q^2 =-
\begin{pmatrix}
0 &  0  &  1 \\
0 &  1 &  0 \\
1 &  0  &  0 \\
\end{pmatrix}
$

$
H = \frac{1}{\sqrt{3}i}
\begin{pmatrix}
1 &  1  &  1 \\
1 &  \omega & \omega^2 \\
1 &  \omega^2  &  \omega \\
\end{pmatrix}
$

Thus, by braiding alone we obtained all the $1$-qutrit classical gates, the generalized Hadamard gate and the gates $Q[i]_3,\; i = 0, 1 , 2$.

Next, we consider the encoding of the $2$-qutrits using the $9$ dimensional subspace $V^{\epsilon\epsilon\epsilon\epsilon}_y \bigotimes V^{\epsilon\epsilon\epsilon\epsilon}_y \subset V^{\epsilon\epsilon\epsilon\epsilon\epsilon\epsilon\epsilon\epsilon}_y$. See Figure \ref{2-qutrit pic}.

\begin{figure}
\includegraphics{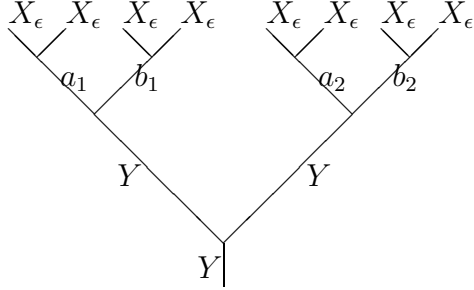}
\caption{2-qutrit model}\label{2-qutrit pic}
\end{figure}

%
%
%
%
%
%

Let $s_1 = \sigma_2\sigma_1\sigma_3\sigma_2,\ \ s_2 = \sigma_4\sigma_3\sigma_5\sigma_4,\ s_3 = \sigma_6\sigma_5\sigma_7\sigma_6. \ $ And let $\bigwedge(Z) = s_1^{-1}s_2^2s_1 s_3^{-1}s_2^2 s_3.$

It can be verified that $\bigwedge(Z),$ restricted to the $9$-dimensional subspace $V^{\epsilon\epsilon\epsilon\epsilon}_y \bigotimes V^{\epsilon\epsilon\epsilon\epsilon}_y \subset V^{\epsilon\epsilon\epsilon\epsilon\epsilon\epsilon\epsilon\epsilon}_y,$ is exactly the Controlled-$Z$ gate when $\{-|YY\rangle, |\mathbf{1}Y\rangle, |Y \mathbf{1}\rangle\}$ is the computational basis for each qutrit $V^{\epsilon\epsilon\epsilon\epsilon}_y.$  By drawing the braids $s_i, i=1,2,3$, it is not hard to be convinced that this Controlled-$Z$ gate has no leakage.  Therefore, our anyonic model is leakage-free.

Recall the definitions from the beginning of Section \ref{univeral gate set}, we see that the $\SUM$ gate is related to the Controlled-$Z$ gate through  the generalized Hadamard gate $H$. Explicitly, we have
$$\SUM = (Id \otimes H)\bigwedge(Z)^{-1}(Id \otimes H^{-1}).$$

Thus, we can obtain the $\SUM$ gate by braiding, since $H$ is already a braiding circuit. We have:

\begin{prop}\label{propbraiding}
By braiding alone, we can construct the classical $1$-qutrit gates, the generalized Hadamard gate, the generalized $\sigma_z$ gates $Q[i]$, and the $\SUM$ gate for our anyonic model $V^{\epsilon\epsilon\epsilon\epsilon}_y$.
\end{prop}

It follows from Theorem \ref{qutrit thm} that we need to find the square roots of $Q[i]$ to make our model universal.  Our solution is to introduce a physically realistic measurement: to determine wether or not the total charge of two anyons is trivial.

\begin{measurement}\label{measure 1}
Let $\M_{\mathbf{1}}=\{\Pi_\mathbf{1},\Pi_\mathbf{1}'\}$ be the projective measurement onto the total charge=$\mathbf{1}$ sector of {\it two} anyons and its complement. Then $\M_\mathbf{1}$ allows us to distinguish between the trivial anyon $\mathbf{1}$ and other anyons; namely, check whether an anyon is trivial or not. Moreover, in a $1$-qutrit model, the state of the second pair after each outcome of the measurement of the first pair is still coherent.
\end{measurement}

Applying Measurement \ref{measure 1} to the first two anyons in the $1$-qutrit model to determine their total charge is equivalent to projecting the state to the subspace $span\{|\mathbf{1}Y\rangle\}$ and its orthogonal complement $span\{-|YY\rangle, |Y \mathbf{1}\rangle\}$.
Since all the $1$-qutrit classical gates can be constructed by braiding, we can also project the state to $span\{-|YY\rangle\}$ and $span\{|\mathbf{1}Y\rangle, |Y \mathbf{1}\rangle\}$. Thus, Measurement \ref{measure 0 12} can be obtained from Measurement \ref{measure 1} and braiding.  It is important to notice that when the total charge of the first two anyons of a qutrit is $Y$, then the total charge of the second pair of anyons is in a coherent superposition of $\mathbf{1}$ and $Y$.

Another method to measure total charge of anyons is interferometric measurement.  It is known that any projective measurement of total charge of anyons can be simulated by interferometric measurements \cite{FL15}.

Therefore, by braiding anyons and Measurement \ref{measure 1}, we can construct the generalized Hadamard gate $H$, the $\SUM$ gate, all the $Q[i]\;'$s, all the $1$-qutrit classical gates and Measuremnt \ref{measure 0 12}. These are exactly the universal gate set in Theorem \ref{qutrit thm2}.

\begin{thm}
In the $\SU(2)_4$ theory, if we use the fusion space of four metaplectic anyons $X_{\epsilon}$ with total charge $Y$ as a $1$-qutrit $($See \ref{1-qutrit pic}$)$, and choose $\{-|YY\rangle, |\mathbf{1}Y\rangle, |Y \mathbf{1}\rangle\}$ as the computational basis, then braiding supplemented with Measurement \ref{measure 1} of two metaplectic anyons forms a universal gate set for our anyonic quantum computation.
\begin{proof}
It follows from Prop. \ref{propbraiding} and Theorem \ref{qutrit thm2}.
\end{proof}
\end{thm}

The most challenging part of our model is to maintain the total charge of many metaplectic anyons to be the meteplectic mode $Y$.  We cannot make the more natural choice of total charge trivial model universal, and provide evidence below that this cannot be done.

\subsection{A model from $\SO(5)_2$} \label{SO(5)}

The $\SO(5)_2$ theory consists of six anyon types: $\{{1}, z, y_1, y_2, {\epsilon}, {\epsilon}'\}.$ We set up a similar model as that for $\SO(3)_2$.

For a $1$-qupit $p=5$, use the model as shown in Figure \ref{1-qupit pic}.

\begin{figure}
\includegraphics{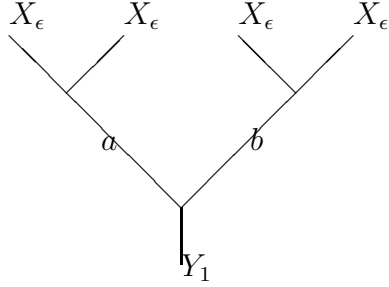}
\caption{1-qupit model}\label{1-qupit pic}
\end{figure}
%
%

The Hilbert space $V^{\epsilon\epsilon\epsilon\epsilon}_{y_1}$ now is $5$-dimensional with the computational basis $\{|Y_2Y_2\rangle, |\mathbf{1}Y_1\rangle, |Y_2Y_1\rangle, |Y_1Y_2\rangle, |Y_1\mathbf{1}\rangle\}$. The representation matrices of the generators of $\B_4$ are:

$\sigma_1 = \frac{1}{i} \left(
\begin{array}{ccccc}
 e^{\frac{2 i \pi }{5}} & 0 & 0 & 0 & 0 \\
 0 & 1 & 0 & 0 & 0 \\
 0 & 0 & e^{\frac{2 i \pi }{5}} & 0 & 0 \\
 0 & 0 & 0 & e^{-\frac{2 i \pi }{5}} & 0 \\
 0 & 0 & 0 & 0 & e^{-\frac{2 i \pi }{5}} \\
\end{array}
\right)
$

$\sigma_2 = \frac{1}{\sqrt{5}i} \left(
\begin{array}{ccccc}
 1 & e^{-\frac{2 i \pi }{5}} & e^{\frac{2 i \pi }{5}} & e^{\frac{2 i \pi }{5}} & e^{-\frac{2 i \pi }{5}} \\
 e^{-\frac{2 i \pi }{5}} & 1 & e^{-\frac{2 i \pi }{5}} & e^{\frac{2 i \pi }{5}} & e^{\frac{2 i \pi }{5}} \\
 e^{\frac{2 i \pi }{5}} & e^{-\frac{2 i \pi }{5}} & 1 & e^{-\frac{2 i \pi }{5}} & e^{\frac{2 i \pi }{5}} \\
 e^{\frac{2 i \pi }{5}} & e^{\frac{2 i \pi }{5}} & e^{-\frac{2 i \pi }{5}} & 1 & e^{-\frac{2 i \pi }{5}} \\
 e^{-\frac{2 i \pi }{5}} & e^{\frac{2 i \pi }{5}} & e^{\frac{2 i \pi }{5}} & e^{-\frac{2 i \pi }{5}} & 1 \\
\end{array}
\right)
$

$\sigma_3 = \frac{1}{i}\left(
\begin{array}{ccccc}
 e^{\frac{2 i \pi }{5}} & 0 & 0 & 0 & 0 \\
 0 & e^{-\frac{2 i \pi }{5}} & 0 & 0 & 0 \\
 0 & 0 & e^{-\frac{2 i \pi }{5}} & 0 & 0 \\
 0 & 0 & 0 & e^{\frac{2 i \pi }{5}} & 0 \\
 0 & 0 & 0 & 0 & 1 \\
\end{array}
\right)
$

The representation $V^{\epsilon\epsilon\epsilon\epsilon}_{y_1}$ of the braid group $\B_4$ is irreducible and the image is the Clifford group, which is isomorphic to $(\mathbb{Z}_5 \times \mathbb{Z}_5) \rtimes \mathrm{SL(2, \mathbb{Z}_5)}$.

Direct calculations lead to the following important gates, up to a phase, from braiding.

The generalized Hadamard gate $H_5:$ $H_5 |j\rangle = \sum\limits_{i=0}^{4} \omega_5^{ij}|i\rangle,\,= \sigma_1^{-1}\sigma_3^{-1}\sigma_2^2\sigma_1^{-1}\sigma_3^{-1}$.

The generalized $Z$-gate $Z:$ $Z|i\rangle = \omega_5^i |i\rangle$ = $\sigma_1\sigma_3^{-1}.$

The generalized $X$-gate $X:$ $X|i\rangle = |i+1\rangle$ = $\sigma_1\sigma_2\sigma_1^{-2}\sigma_3^2(\sigma_1\sigma_2)^{-1}.$

The multiplication gate $M[k]:$ $M[k] |i\rangle = |ki\rangle, \, k = 2,3,4$. These gates are realized by $\sigma_1^2\sigma_2^{-2}\sigma_1^{-1}\sigma_3^{-1}\sigma_2\sigma_1, \sigma_1^2\sigma_2^{-1}\sigma_1\sigma_3\sigma_2^2\sigma_3,$ and $\sigma_1\sigma_2\sigma_1\sigma_3\sigma_2\sigma_1.$

$\\$

The gates $X$ and $M[k]$ are classical $1$-qupit gates, i.e., permutation matrices. Identifying a permutation matrix with the permutation that it represents in the permutation group $\Sym_5$, we see that $X, M[k]$ generate a maximal subgroup of $\Sym_5$ with order $20$, which is isomorphic to $\mathbb{Z}_5 \rtimes \mathbb{Z}_4$. Moreover, this maximal subgroup contains all the $1$-qupit classical gates obtained from braiding. Since this subgroup is maximal, any classical gate out of the subgroup is enough to produce all the $1$-qupit classical gates.
We speculate that Measurement \ref{measure 1} would help produce an extra classical gate.

The $2$-qupit encoding is also analogous to the one we used in $\SO(3)_2$. See Figure \ref{2-qupit pic}.

\begin{figure}
\includegraphics{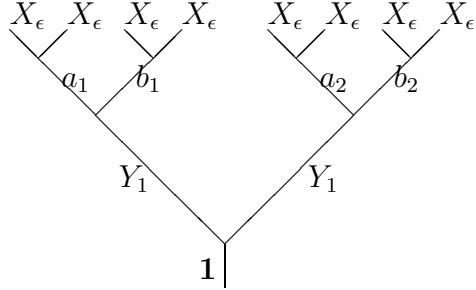}
\caption{2-qupit model}\label{2-qupit pic}
\end{figure}
%
%
%
%
%
%

We obtain the Controlled-$Z$ gate $\bigwedge(Z)$ by the same braiding as we did in Subsection \ref{SO(3)}. Here $\bigwedge(Z)|i,j\rangle = \omega_5^{ij}|i,j\rangle.$ And again, the $\SUM$ gate is obtained by conjugating $\bigwedge(Z)$ by the generalized Hadamard $H$.

$\\$

\begin{prop}
The gates that can be constructed from braiding include the generalized Hadamard $H$, $20$ $1$-qupit classical gates generated by the generalized $X$ gate and multiplication gates $M[k]\;'$s, the generalized $Z$ gate and the $\SUM$ gate. In view of Corollary \ref{qudit cor}, we need an extra $1$-qupit classical gate and some gate $Q[i]$ to make this model universal.
\end{prop}

Suppose we have all the $1$-qupit classical gates, then clearly by using $\sigma_1$ and classical gates, we can obtain the gates $R[i,j,k] = (Q[i]Q[j]^{-1})^k, i \neq j, k= 1,2,3,4$. For example,

$R[1,2,k] =
\begin{pmatrix}
\omega^k & 0 & 0 & 0 & 0\\
0 & \omega^{-k} & 0 & 0 & 0\\
0 & 0 & 1 & 0 & 0\\
0 & 0 & 0 & 1 & 0\\
0 & 0 & 0 & 0 & 1\\

\end{pmatrix}
$

Let $X[i,j,k] = HR[i,j,k]H^{-1}R[i,j,k]^{-1},$.

It can be shown that $X[i,j,k]$ is of infinite order, and for fixed $i,j$, the four matrices $\{X[i,j,k], k= 1,2,3,4\}$ fix some $1$-dimensional subspace and act irreducibly on the $4$-dimensional orthogonal complement. For example, $\{X[1,2,k], k= 1,2,3,4\}$ fix the vector $\omega^{-1} |2\rangle + \frac{\sqrt{5}+1}{2}\omega^2 |3\rangle + |4\rangle.$ If one can show $\{X[i,j,k], k= 1,2,3,4\}$ generate a dense subgroup of the unitary group of the $4$-dimensional complement for some $i\neq j$, then it is straightforward to prove the gate set generate a dense subgroup of $\SU(5)$ by Lemma \ref{stablizer universal}. We did not succeed in showing this either.

\subsection{Other models with $\SU(2)_4$}

There are at least $4$ obvious anyonic quantum computing models with $\SU(2)_4$ anyons.  Besides the universal model that we studied, three others are the qubit model $V^{1111}_0$, the qutrit model $V^{2222}_0$, and the qubit model $V^{1221}_0$.  The computational power of the corresponding models $V^{2222}_0$ and $V^{1221}_0$ in the Jones-Kauffman version of $SU(2)_4$ is analyzed in \cite{Gates}.  We conjecture that the model $V^{1111}_0$, shown in Figure \ref{1-qubit pic}, with measurements of total charges of metaplectic anyons is Ising like, i.e., braidings and such measurements can be simulated classically efficiently.

\begin{figure}
\includegraphics{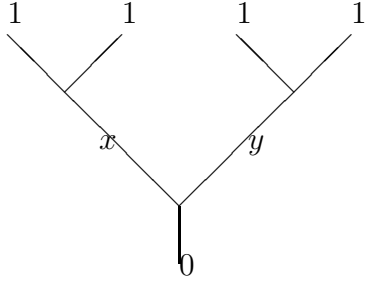}
\caption{1-qubit model}\label{1-qubit pic}
\end{figure}
%
%

The Hilbert space $V^{1111}_0$ is $2$-dimensional with basis $\{|00\rangle, |22\rangle\}.$ Under this basis, the $\sigma_i \,'s$ have the following matrices:

$\sigma_1 = \sigma_3 = \gamma
\begin{pmatrix}
\omega &  0  \\
0 &  1  \\
\end{pmatrix}
\qquad
\sigma_2 = \gamma \bar{\omega}
\begin{pmatrix}
-\frac{1}{2} + \frac{\sqrt{3}i}{6}   & \frac{\sqrt{6}i}{3}\\
\frac{\sqrt{6}i}{3}                 & -\frac{1}{2} - \frac{\sqrt{3}i}{6} \\
\end{pmatrix},
$ for some phase $\gamma$.

Up to normalization, this representation is the same as one component of $V^{DDDD}_B$ in $D(S_3)$, where the two components of $V^{DDDD}_B$ are isomorphic. See Appendix B.2.2 in \cite{cui2014universal}.

These matrices generate a group of size $24$ which is isomorphic to $\textrm{SL}(2, \mathbb{F}_3)$. Modulo the center, we get the even permutations $\A_4$.

Similarly, for the $2$-qubit encoding as that in Figure \ref{2-qubit pic}.

\begin{figure}
\includegraphics{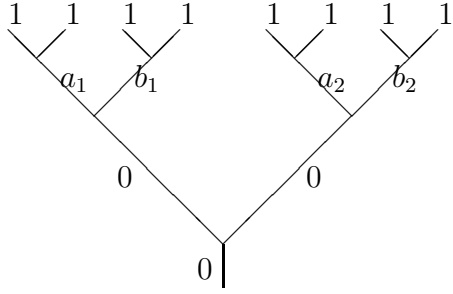}
\caption{2-qubit model}\label{2-qubit pic}
\end{figure}

%
%
%
%
%
%

We use $|0;a_1b_1\rangle \otimes |0;a_2b_2\rangle$ to denote the state in Figure \ref{2-qubit pic}.

Then the same braiding as before gives the following transformation: \\
$|0;00\rangle \otimes |0;00\rangle \mapsto |0;00\rangle \otimes |0;00\rangle,$ \\
$|0;00\rangle \otimes |0;22\rangle \mapsto |0;00\rangle \otimes |0;22\rangle,$  \\
$|0;22\rangle \otimes |0;00\rangle \mapsto |0;22\rangle \otimes |0;00\rangle,$\\
$|0;22\rangle \otimes |0;22\rangle \mapsto -\frac{1}{2}|0;22\rangle \otimes |0;22\rangle + \frac{\sqrt{3}i}{2}|4;22\rangle \otimes |4;22\rangle$\\

Thus, projecting out the charge $4$, we obtain the Controlled-$\sigma_z$ gate $\bigwedge(\sigma_z).$ But when the state is actually projected onto the charge $4$ part, the state is destroyed and the whole computational process has to start over again.

\appendix

\section{$F$-Matrices and $R$-symbols for $\SU(2)_4$} \label{6jSO(3)}

We order the labels $\{0,1,2,3,4\}$ in increasing order when we arrange the entries of the $F$-matrices. The conventions that we used for $F$-matrices and $R$-symbols are in \cite{cui2014universal}. We drop all trivial $F^{abc}_d=1$ when $a,b,c,d$ is admissible, in particular when one of $a ,\; b$ or $c$ is trivial.

\begin{itemize}

\item $\left(
\begin{array}{c}
 -1 \\
\end{array}
\right)\: \: \textrm{for} $

$F^{114}_{4}, F^{123}_{4}, F^{124}_{3}, F^{132}_{4}, F^{133}_{3}, F^{134}_{2}, F^{141}_{4}, F^{142}_{3}, F^{143}_{2}, F^{144}_{1}, F^{213}_{4}, F^{214}_{3}  $

$F^{222}_{4}, F^{224}_{2}, F^{231}_{4}, F^{234}_{1}, F^{241}_{3}, F^{242}_{2}, F^{243}_{1}, F^{312}_{4}, F^{313}_{3}, F^{314}_{2}, F^{321}_{4}, F^{324}_{1}  $

$F^{331}_{3}, F^{333}_{1}, F^{334}_{4}, F^{341}_{2}, F^{342}_{1}, F^{343}_{4}, F^{344}_{3}, F^{411}_{4}, F^{412}_{3}, F^{413}_{2}, F^{414}_{1}, F^{421}_{3}  $

$F^{422}_{2}, F^{423}_{1}, F^{431}_{2}, F^{432}_{1}, F^{433}_{4}, F^{434}_{3}, F^{441}_{1}, F^{443}_{3}  $

\item $\left(
\begin{array}{cc}
 -\frac{1}{\sqrt{3}} & \sqrt{\frac{2}{3}} \\
 \sqrt{\frac{2}{3}} & \frac{1}{\sqrt{3}} \\
\end{array}
\right)\: \: \textrm{for} \: \: $ $F^{111}_{1}, F^{131}_{3}, F^{313}_{1}, F^{333}_{3}  $

\item $\left(
\begin{array}{cc}
 -\frac{1}{\sqrt{2}} & \frac{1}{\sqrt{2}} \\
 \frac{1}{\sqrt{2}} & \frac{1}{\sqrt{2}} \\
\end{array}
\right)\: \: \textrm{for} $

$F^{112}_{2}, F^{122}_{1}, F^{122}_{3}, F^{132}_{2}, F^{211}_{2}, F^{213}_{2}, F^{221}_{1}, F^{221}_{3}, F^{223}_{1}, F^{231}_{2}, F^{312}_{2}, F^{322}_{1}  $

\item $\left(
\begin{array}{cc}
 -\sqrt{\frac{2}{3}} & \frac{1}{\sqrt{3}} \\
 \frac{1}{\sqrt{3}} & \sqrt{\frac{2}{3}} \\
\end{array}
\right)\: \: \textrm{for} \: \:$ $F^{113}_{3}, F^{133}_{1}, F^{311}_{3}, F^{331}_{1}  $

\item $\left(
\begin{array}{cc}
 -\frac{1}{2} & \frac{\sqrt{3}}{2} \\
 \frac{\sqrt{3}}{2} & \frac{1}{2} \\
\end{array}
\right)\: \: \textrm{for} \: \:$ $F^{121}_{2}, F^{212}_{1}  $

\item $\left(
\begin{array}{cc}
 -\frac{\sqrt{3}}{2} & \frac{1}{2} \\
 \frac{1}{2} & \frac{\sqrt{3}}{2} \\
\end{array}
\right)\: \: \textrm{for} \: \: $ $F^{123}_{2}, F^{212}_{3}, F^{232}_{1}, F^{321}_{2}  $

\item $\left(
\begin{array}{cc}
 \frac{1}{\sqrt{2}} & -\frac{1}{\sqrt{2}} \\
 -\frac{1}{\sqrt{2}} & -\frac{1}{\sqrt{2}} \\
\end{array}
\right)\: \: \textrm{for} \: \: $ $F^{223}_{3}, F^{233}_{2}, F^{322}_{3}, F^{332}_{2}  $

\item $\left(
\begin{array}{cc}
 \frac{1}{2} & -\frac{\sqrt{3}}{2} \\
 -\frac{\sqrt{3}}{2} & -\frac{1}{2} \\
\end{array}
\right)\: \: \textrm{for} \: \:$ $F^{232}_{3}, F^{323}_{2}  $

\item $\left(
\begin{array}{ccc}
 \frac{1}{2} & -\frac{1}{\sqrt{2}} & \frac{1}{2} \\
 -\frac{1}{\sqrt{2}} & 0 & \frac{1}{\sqrt{2}} \\
 \frac{1}{2} & \frac{1}{\sqrt{2}} & \frac{1}{2} \\
\end{array}
\right)\: \: \textrm{for}\: \: $ $F^{222}_{2}  $
\end{itemize}

\subsection{$R$-symbols}

\begin{itemize}

\item $1\: \: \textrm{for} \: \: R^{00}_{0}, R^{01}_{1}, R^{02}_{2}, R^{03}_{3}, R^{04}_{4}, R^{10}_{1}, R^{20}_{2}, R^{30}_{3}, R^{40}_{4}, R^{44}_{0}$

\item $e^{\frac{3 i \pi }{4}}\: \: \textrm{for} \: \: R^{11}_{0}$

\item $e^{\frac{i \pi }{12}}\: \: \textrm{for} \: \: R^{11}_{2}$

\item $e^{\frac{2 i \pi }{3}}\: \: \textrm{for} \: \: R^{12}_{1}, R^{21}_{1}, R^{22}_{2}, R^{23}_{3}, R^{32}_{3}$

\item $e^{\frac{i \pi }{6}}\: \: \textrm{for} \: \: R^{12}_{3}, R^{21}_{3}$

\item $e^{\frac{7 i \pi }{12}}\: \: \textrm{for} \: \: R^{13}_{2}, R^{31}_{2}$

\item $e^{\frac{i \pi }{4}}\: \: \textrm{for} \: \: R^{13}_{4}, R^{31}_{4}$

\item $i\: \: \textrm{for} \: \: R^{14}_{3}, R^{41}_{3}$

\item $e^{-\frac{2 i \pi }{3}}\: \: \textrm{for} \: \: R^{22}_{0}$

\item $e^{\frac{i \pi }{3}}\: \: \textrm{for} \: \: R^{22}_{4}$

\item $e^{-\frac{5 i \pi }{6}}\: \: \textrm{for} \: \: R^{23}_{1}, R^{32}_{1}$

\item $-1\: \: \textrm{for} \: \: R^{24}_{2}, R^{42}_{2}$

\item $e^{-\frac{i \pi }{4}}\: \: \textrm{for} \: \: R^{33}_{0}$

\item $e^{-\frac{11 i \pi }{12}}\: \: \textrm{for} \: \: R^{33}_{2}$

\item $-i\: \: \textrm{for} \: \: R^{34}_{1}, R^{43}_{1}$
\end{itemize}

\section{$F$-Matrices and $R$-symbols for $\SO(5)_2$}\label{6jSO(5)}

Here we list all the $6j$ symbols and some of the $R$-symbols that we need in this paper for the theory $\SO(5)_2.$ Again we omit the trivial  $F^{abc}_d$.  We arrange the label set in the order $\{1,z, y_1,y_2, \epsilon, \epsilon'\}$ in the following.

Let $h = \sqrt{10-2 \sqrt{5}}, k = \sqrt{10+2 \sqrt{5}}.$

\begin{itemize}

\item $\left(
\begin{array}{c}
 -1 \\
\end{array}
\right)\: \: \textrm{for} $

$F^{z y_1 y_1 }_{y_2 }, F^{z y_1 y_2 }_{y_1 }, F^{z y_2 y_1 }_{y_1 }, F^{z y_2 y_1 }_{y_2 }, F^{z \epsilon z }_{\epsilon }, F^{z \epsilon y_1 }_{\epsilon' }, F^{z \epsilon y_2 }_{\epsilon' }, F^{z \epsilon' z }_{\epsilon' }, F^{z \epsilon' y_1 }_{\epsilon }  $

$F^{z \epsilon' y_2 }_{\epsilon }, F^{y_1 z y_1 }_{y_2 }, F^{y_1 z y_2 }_{y_1 }, F^{y_1 y_1 z }_{y_2 }, F^{y_1 y_1 y_2 }_{z }, F^{y_1 y_2 z }_{y_1 }, F^{y_1 y_2 z }_{y_2 }, F^{y_1 y_2 y_1 }_{z }, F^{y_1 \epsilon z }_{\epsilon' }  $

$F^{y_1 \epsilon' z }_{\epsilon }, F^{y_2 z y_1 }_{y_1 }, F^{y_2 z y_2 }_{y_1 }, F^{y_2 y_1 z }_{y_1 }, F^{y_2 y_1 y_1 }_{z }, F^{y_2 y_1 y_2 }_{z }, F^{y_2 \epsilon z }_{\epsilon' }, F^{y_2 \epsilon' z }_{\epsilon }, F^{\epsilon z \epsilon }_{z }  $

$F^{\epsilon z \epsilon' }_{y_1 }, F^{\epsilon z \epsilon' }_{y_2 }, F^{\epsilon y_1 \epsilon' }_{z }, F^{\epsilon y_2 \epsilon' }_{z }, F^{\epsilon' z \epsilon }_{y_1 }, F^{\epsilon' z \epsilon }_{y_2 }, F^{\epsilon' z \epsilon' }_{z }, F^{\epsilon' y_1 \epsilon }_{z }, F^{\epsilon' y_2 \epsilon }_{z }  $

\item $\left(
\begin{array}{cc}
 \frac{1}{\sqrt{2}} & -\frac{1}{\sqrt{2}} \\
 \frac{1}{\sqrt{2}} & \frac{1}{\sqrt{2}} \\
\end{array}
\right)\: \: \textrm{for} $

$F^{y_1 y_1 y_2 }_{y_2 }, F^{y_1 y_1 \epsilon }_{\epsilon' }, F^{y_1 y_1 \epsilon' }_{\epsilon }, F^{y_1 \epsilon \epsilon }_{y_2 }, F^{y_2 y_2 y_1 }_{y_1 }, F^{y_2 \epsilon \epsilon }_{y_1 }, F^{\epsilon y_1 y_2 }_{\epsilon }, F^{\epsilon y_2 y_1 }_{\epsilon }, F^{\epsilon \epsilon' y_1 }_{y_1 }  $

$F^{\epsilon' \epsilon y_1 }_{y_1 }  $

\item $\left(
\begin{array}{cc}
 \frac{1}{\sqrt{2}} & \frac{1}{\sqrt{2}} \\
 \frac{1}{\sqrt{2}} & -\frac{1}{\sqrt{2}} \\
\end{array}
\right)\: \: \textrm{for} $

$F^{y_1 y_1 \epsilon }_{\epsilon }, F^{y_1 y_1 \epsilon' }_{\epsilon' }, F^{y_1 \epsilon \epsilon }_{y_1 }, F^{y_1 \epsilon' \epsilon' }_{y_1 }, F^{y_2 y_2 \epsilon }_{\epsilon }, F^{y_2 y_2 \epsilon }_{\epsilon' }, F^{y_2 y_2 \epsilon' }_{\epsilon }, F^{y_2 y_2 \epsilon' }_{\epsilon' }, F^{y_2 \epsilon \epsilon }_{y_2 }  $

$F^{y_2 \epsilon \epsilon' }_{y_2 }, F^{y_2 \epsilon' \epsilon }_{y_2 }, F^{y_2 \epsilon' \epsilon' }_{y_2 }, F^{\epsilon y_1 y_1 }_{\epsilon }, F^{\epsilon y_2 y_2 }_{\epsilon }, F^{\epsilon y_2 y_2 }_{\epsilon' }, F^{\epsilon \epsilon y_1 }_{y_1 }, F^{\epsilon \epsilon y_2 }_{y_2 }, F^{\epsilon \epsilon' y_2 }_{y_2 }  $

$F^{\epsilon' y_1 y_1 }_{\epsilon' }, F^{\epsilon' y_2 y_2 }_{\epsilon }, F^{\epsilon' y_2 y_2 }_{\epsilon' }, F^{\epsilon' \epsilon y_2 }_{y_2 }, F^{\epsilon' \epsilon' y_1 }_{y_1 }, F^{\epsilon' \epsilon' y_2 }_{y_2 }  $

\item $\left(
\begin{array}{cc}
 0 & 1 \\
 1 & 0 \\
\end{array}
\right)\: \: \textrm{for} $   $F^{y_1 y_2 y_1 }_{y_2 }, F^{y_2 y_1 y_2 }_{y_1 }  $

\item $\left(
\begin{array}{cc}
 \frac{1}{\sqrt{2}} & \frac{1}{\sqrt{2}} \\
 -\frac{1}{\sqrt{2}} & \frac{1}{\sqrt{2}} \\
\end{array}
\right)\: \: \textrm{for} $

$F^{y_1 y_2 y_2 }_{y_1 }, F^{y_1 y_2 \epsilon }_{\epsilon }, F^{y_1 \epsilon \epsilon' }_{y_1 }, F^{y_1 \epsilon' \epsilon }_{y_1 }, F^{y_2 y_1 y_1 }_{y_2 }, F^{y_2 y_1 \epsilon }_{\epsilon }, F^{\epsilon y_1 y_1 }_{\epsilon' }, F^{\epsilon \epsilon y_1 }_{y_2 }, F^{\epsilon \epsilon y_2 }_{y_1 }  $

$F^{\epsilon' y_1 y_1 }_{\epsilon }  $

\item $\left(
\begin{array}{cc}
 -\frac{1}{\sqrt{2}} & \frac{1}{\sqrt{2}} \\
 \frac{1}{\sqrt{2}} & \frac{1}{\sqrt{2}} \\
\end{array}
\right)\: \: \textrm{for} $

$F^{y_1 y_2 \epsilon }_{\epsilon' }, F^{y_1 \epsilon' \epsilon }_{y_2 }, F^{y_2 y_1 \epsilon' }_{\epsilon }, F^{y_2 \epsilon \epsilon' }_{y_1 }, F^{\epsilon y_2 y_1 }_{\epsilon' }, F^{\epsilon \epsilon' y_1 }_{y_2 }, F^{\epsilon' y_1 y_2 }_{\epsilon }, F^{\epsilon' \epsilon y_2 }_{y_1 }  $

\item $\left(
\begin{array}{cc}
 -\frac{1}{\sqrt{2}} & -\frac{1}{\sqrt{2}} \\
 \frac{1}{\sqrt{2}} & -\frac{1}{\sqrt{2}} \\
\end{array}
\right)\: \: \textrm{for} $    $F^{y_1 y_2 \epsilon' }_{\epsilon }, F^{y_2 y_1 \epsilon }_{\epsilon' }, F^{\epsilon \epsilon' y_2 }_{y_1 }, F^{\epsilon' \epsilon y_1 }_{y_2 }  $

\item $\left(
\begin{array}{cc}
 \frac{1}{\sqrt{2}} & -\frac{1}{\sqrt{2}} \\
 -\frac{1}{\sqrt{2}} & -\frac{1}{\sqrt{2}} \\
\end{array}
\right)\: \: \textrm{for} $

$F^{y_1 y_2 \epsilon' }_{\epsilon' }, F^{y_1 \epsilon' \epsilon' }_{y_2 }, F^{y_2 y_1 \epsilon' }_{\epsilon' }, F^{y_2 \epsilon' \epsilon' }_{y_1 }, F^{\epsilon' y_1 y_2 }_{\epsilon' }, F^{\epsilon' y_2 y_1 }_{\epsilon' }, F^{\epsilon' \epsilon' y_1 }_{y_2 }, F^{\epsilon' \epsilon' y_2 }_{y_1 }  $

\item $\frac{1}{4}\left(
\begin{array}{cc}
 -\frac{\sqrt{5}}{10}k^2 & h \\
 h  & \frac{\sqrt{5}}{10}k^2 \\
\end{array}
\right)\: \: \textrm{for} $   $F^{y_1 \epsilon y_1 }_{\epsilon }, F^{\epsilon y_1 \epsilon }_{y_1 }  $

\item $\frac{1}{4}\left(
\begin{array}{cc}
 h & \frac{\sqrt{5}}{10}k^2 \\
 \frac{\sqrt{5}}{10}k^2 & -h \\
\end{array}
\right)\: \: \textrm{for} $   $F^{y_1 \epsilon y_1 }_{\epsilon' }, F^{y_1 \epsilon' y_1 }_{\epsilon }, F^{\epsilon y_1 \epsilon' }_{y_1 }, F^{\epsilon' y_1 \epsilon }_{y_1 }  $

\item $\frac{1}{4}\left(
\begin{array}{cc}
 \frac{\sqrt{5}}{10}h^2 & k \\
  k & -\frac{\sqrt{5}}{10}h^2 \\
\end{array}
\right)\: \: \textrm{for} $   $F^{y_1 \epsilon y_2 }_{\epsilon }, F^{y_2 \epsilon y_1 }_{\epsilon }, F^{\epsilon y_1 \epsilon }_{y_2 }, F^{\epsilon y_2 \epsilon }_{y_1 }  $

\item $\frac{1}{4}\left(
\begin{array}{cc}
 k & -\frac{\sqrt{5}}{10}h^2 \\
 -\frac{\sqrt{5}}{10}h^2 & -\frac{\sqrt{5}}{20}hk^2 \\
\end{array}
\right)\: \: \textrm{for} $

$F^{y_1 \epsilon y_2 }_{\epsilon' }, F^{y_1 \epsilon' y_2 }_{\epsilon }, F^{y_2 \epsilon y_1 }_{\epsilon' }, F^{y_2 \epsilon' y_1 }_{\epsilon }, F^{\epsilon y_1 \epsilon' }_{y_2 }, F^{\epsilon y_2 \epsilon' }_{y_1 }, F^{\epsilon' y_1 \epsilon }_{y_2 }, F^{\epsilon' y_2 \epsilon }_{y_1 }  $

\item $\left(
\begin{array}{cc}
 -\frac{1}{\sqrt{2}} & \frac{1}{\sqrt{2}} \\
 -\frac{1}{\sqrt{2}} & -\frac{1}{\sqrt{2}} \\
\end{array}
\right)\: \: \textrm{for} $  $F^{y_1 \epsilon \epsilon' }_{y_2 }, F^{y_2 \epsilon' \epsilon }_{y_1 }, F^{\epsilon y_1 y_2 }_{\epsilon' }, F^{\epsilon' y_2 y_1 }_{\epsilon }  $

\item $\frac{1}{4}\left(
\begin{array}{cc}
 \frac{\sqrt{5}}{10}k^2 & -h \\
 -h & -\frac{\sqrt{5}}{10}k^2 \\
\end{array}
\right)\: \: \textrm{for} $  $F^{y_1 \epsilon' y_1 }_{\epsilon' }, F^{\epsilon' y_1 \epsilon' }_{y_1 }  $

\item $-\frac{\sqrt{5}h}{40}\left(
\begin{array}{cc}
 h & \frac{k^2}{2} \\
 \frac{k^2}{2} & -h \\
\end{array}
\right)\: \: \textrm{for} $   $F^{y_1 \epsilon' y_2 }_{\epsilon' }, F^{y_2 \epsilon' y_1 }_{\epsilon' }, F^{\epsilon' y_1 \epsilon' }_{y_2 }, F^{\epsilon' y_2 \epsilon' }_{y_1 }  $

\item $\frac{1}{4}\left(
\begin{array}{cc}
 -\frac{\sqrt{5}}{10}k^2 & -h \\
 -h & \frac{\sqrt{5}}{10}k^2 \\
\end{array}
\right)\: \: \textrm{for} $  $F^{y_2 \epsilon y_2 }_{\epsilon }, F^{\epsilon y_2 \epsilon }_{y_2 }  $

\item $\frac{1}{4}\left(
\begin{array}{cc}
 -h & \frac{\sqrt{5}}{10}k^2 \\
 \frac{\sqrt{5}}{10}k^2 & h \\
\end{array}
\right)\: \: \textrm{for} $   $F^{y_2 \epsilon y_2 }_{\epsilon' }, F^{y_2 \epsilon' y_2 }_{\epsilon }, F^{\epsilon y_2 \epsilon' }_{y_2 }, F^{\epsilon' y_2 \epsilon }_{y_2 }  $

\item $\frac{1}{4}\left(
\begin{array}{cc}
 \frac{\sqrt{5}}{10}k^2 & h \\
 h & -\frac{\sqrt{5}}{10}k^2 \\
\end{array}
\right)\: \: \textrm{for} $   $F^{y_2 \epsilon' y_2 }_{\epsilon' }, F^{\epsilon' y_2 \epsilon' }_{y_2 }  $

\item $\frac{\sqrt{5}}{10}\left(
\begin{array}{cc}
 h & k \\
 k & -h \\
\end{array}
\right)\: \: \textrm{for} $   $F^{\epsilon \epsilon \epsilon }_{\epsilon' }, F^{\epsilon \epsilon \epsilon' }_{\epsilon }, F^{\epsilon \epsilon' \epsilon }_{\epsilon }, F^{\epsilon' \epsilon \epsilon }_{\epsilon }  $

\item $-\frac{\sqrt{5}h}{10}\left(
\begin{array}{cc}
 1 & \frac{\sqrt{5}}{20}k^2 \\
 \frac{\sqrt{5}}{20}k^2 & -1 \\
\end{array}
\right)\: \: \textrm{for} $  $F^{\epsilon \epsilon' \epsilon' }_{\epsilon' }, F^{\epsilon' \epsilon \epsilon' }_{\epsilon' }, F^{\epsilon' \epsilon' \epsilon }_{\epsilon' }, F^{\epsilon' \epsilon' \epsilon' }_{\epsilon }  $

\item $\left(
\begin{array}{ccc}
 \frac{1}{2} & \frac{1}{2} & \frac{1}{\sqrt{2}} \\
 \frac{1}{2} & \frac{1}{2} & -\frac{1}{\sqrt{2}} \\
 \frac{1}{\sqrt{2}} & -\frac{1}{\sqrt{2}} & 0 \\
\end{array}
\right)\: \: \textrm{for} $  $F^{y_1 y_1 y_1 }_{y_1 }, F^{y_2 y_2 y_2 }_{y_2 }  $

\item $\frac{1}{\sqrt{5}}\left(
\begin{array}{ccc}
 1 & \sqrt{2} & \sqrt{2} \\
 \sqrt{2} & -\frac{\sqrt{5}+1}{2} & \frac{\sqrt{5}-1}{2} \\
 \sqrt{2} & \frac{\sqrt{5}-1}{2} & -\frac{\sqrt{5}+1}{2} \\
\end{array}
\right)\: \: \textrm{for} $  $F^{\epsilon \epsilon \epsilon }_{\epsilon }, F^{\epsilon' \epsilon' \epsilon' }_{\epsilon' }  $

\item $\frac{1}{\sqrt{5}}\left(
\begin{array}{ccc}
 1 & -\sqrt{2} & -\sqrt{2} \\
\sqrt{2} & \frac{\sqrt{5}+1}{2} & -\frac{\sqrt{5}-1}{2} \\
 \sqrt{2} & -\frac{\sqrt{5}-1}{2} & \frac{\sqrt{5}+1}{2} \\
\end{array}
\right)\: \: \textrm{for} $  $F^{\epsilon \epsilon \epsilon' }_{\epsilon' }, F^{\epsilon' \epsilon' \epsilon }_{\epsilon }  $

\item $\frac{1}{\sqrt{5}}\left(
\begin{array}{ccc}
 -1 & \sqrt{2} & \sqrt{2} \\
\sqrt{2} & \frac{\sqrt{5}+1}{2} & -\frac{\sqrt{5}-1}{2} \\
 \sqrt{2} & -\frac{\sqrt{5}-1}{2} & \frac{\sqrt{5}+1}{2} \\
\end{array}
\right)\: \: \textrm{for} $  $F^{\epsilon \epsilon' \epsilon }_{\epsilon' }, F^{\epsilon' \epsilon \epsilon' }_{\epsilon }  $

\item $\frac{1}{\sqrt{5}}\left(
\begin{array}{ccc}
 1 & \sqrt{2} & \sqrt{2} \\
-\sqrt{2} & \frac{\sqrt{5}+1}{2} & -\frac{\sqrt{5}-1}{2} \\
 -\sqrt{2} & -\frac{\sqrt{5}-1}{2} & \frac{\sqrt{5}+1}{2} \\
\end{array}
\right)\: \: \textrm{for} $  $F^{\epsilon \epsilon' \epsilon' }_{\epsilon }, F^{\epsilon' \epsilon \epsilon }_{\epsilon' }  $
\end{itemize}

\subsection{$R$-symbols}

We only list the $R$-symbols that are enough for our computational purpose. These data are from \cite{hastings2013metaplectic}.

$R^{y_1y_1}_1 = e^{\frac{6 \pi i}{5}}, \, \, R^{y_1y_1}_z = e^{\frac{\pi i}{5}}, \,\, R^{y_1y_1}_{y_2} = e^{\frac{4 \pi i}{5}}$

$R^{\epsilon \epsilon}_{1} = -i, \,\, R^{\epsilon \epsilon}_{y_1} = e^{\frac{11 \pi i}{10}}, \, \,R^{\epsilon \epsilon}_{y_2} = e^{\frac{- \pi i}{10} }$

\section{Matrices of the generators of $\B_4$ in a $1$-qupit model} \label{braid matrix}
In this appendix, we show how to compute the matrices of the three generators $\sigma_1, \sigma_2,\sigma_3$ of $\B_4$ in a $1$-qupit model $V^{aaaa}_b$ with computational basis $\{|x_i y_i \rangle : 1 \leq i \leq p, p \geq 2 \}$, where $a,b$ are two anyon types in a unitary modular category $\Cat$. The data we will use are some of the $F$-matrices and $R$-matrices, which are defined in Figure \ref{FMatrix} and \ref{RMatrix}, where $F_{d;nm}^{abc}$ is the $(n,m)$-entry of the matrix $F_{d}^{abc}$, and $R^{ba}_{c}$ is a $1 \times 1$ $R$-matrix. For more detailed explanations, see Section $2$ of \cite{cui2014universal}, or \cite{Wang}.

\begin{figure}
\includegraphics{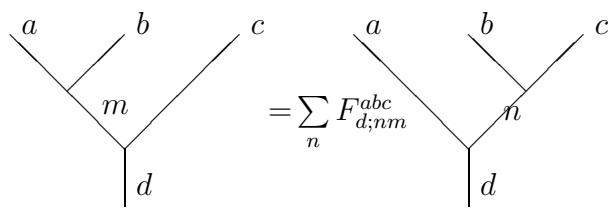}
\caption{Definition of the $F$-matrix $F^{abc}_{d;nm}$ }\label{FMatrix}
\end{figure}

%
%
%
%

%
%
%
%

\begin{figure}
\includegraphics{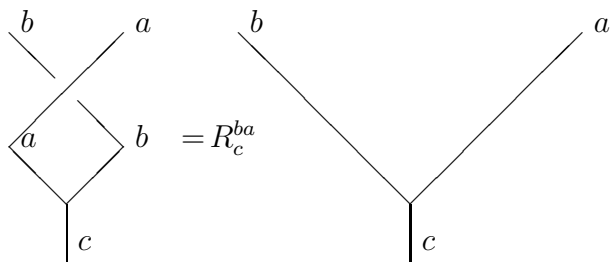}
\caption{Definition of the $R$-matrix $R^{ba}_c$}\label{RMatrix}
\end{figure}

See Figure \ref{1-qupit general} for the $1$-qupit model $V^{aaaa}_b$.


\begin{figure}
\includegraphics{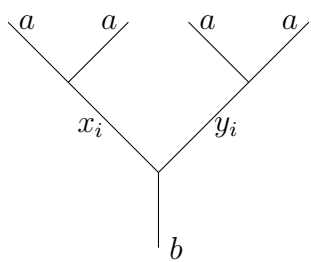}
\caption{General 1-qupit model}\label{1-qupit general}
\end{figure}

Then the braiding pictures of $\sigma_1, \sigma_2, \sigma_3$ are given as shown in Figure \ref{braiding pic}. If we rewrite the braiding pictures in the computational basis, we get the matrices of the corresponding generators. Explicit illustrations are shown below.
%
%

\begin{figure}
\includegraphics{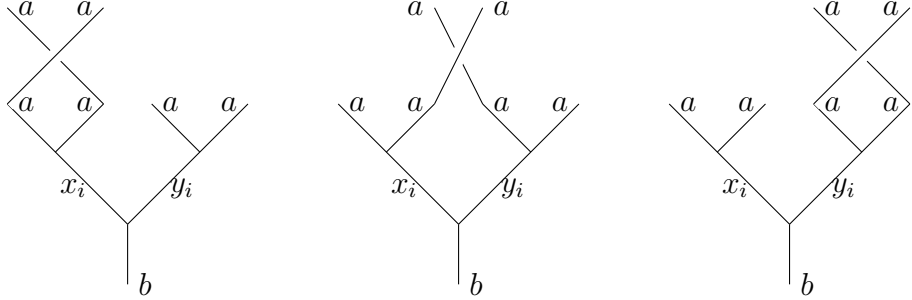}
\caption{Braiding pictures for $\sigma_1, \sigma_2, \sigma_3$} \label{braiding pic}
\end{figure}

By the definition of $R$-matrices in Figure \ref{RMatrix}, the braiding picture for $\sigma_1$ can be written as that shown in Figure \ref{braiding sigma1}. Thus $\sigma_1$ is always the diagonal matrix with $(i,i)$-entry $R^{aa}_{x_i}$. Similarly, $\sigma_3$ is also a diagonal matrix with $(i,i)$-entry $R^{aa}_{y_i}$. The calculation of $\sigma_2$ is much more complicated as it involves change of bases using $F$-matrices. See figure \ref{braiding sigma2} for the illustrations, where ${F^{-1}}^{aaa}_{c;dy_i}$ is the $(d,y_i)$-entry of the inverse of the matrix $F^{aaa}_{c}$, and in the last picture of the equations, the pair $|fe\rangle$ could only be one of the basis elements in $\{|x_jy_j\rangle: 1 \leq j \leq p\}$. Let $|fe\rangle = |x_jy_j\rangle$, then the $(j,i)$-entry of $\sigma_2$ is given by the following:
\begin{equation}
(\sigma_2)_{j,i} = \sum\limits_{c,d} F^{aay_i}_{b;cx_i}{F^{-1}}^{aaa}_{c;dy_i}R^{aa}_dF^{aaa}_{c;y_jd}{F^{-1}}^{aay_j}_{b;x_jc}
\end{equation}

%
%

\begin{figure}
\includegraphics{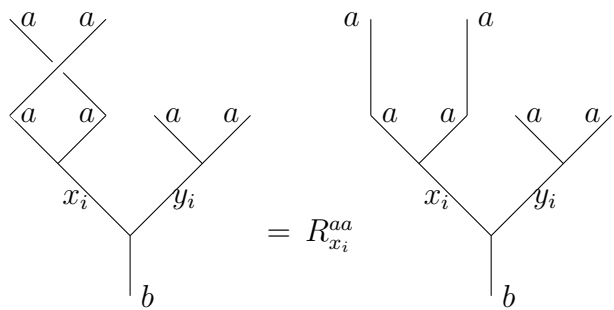}
\caption{Calculating $\sigma_1$} \label{braiding sigma1}
\end{figure}

\begin{figure}
\includegraphics{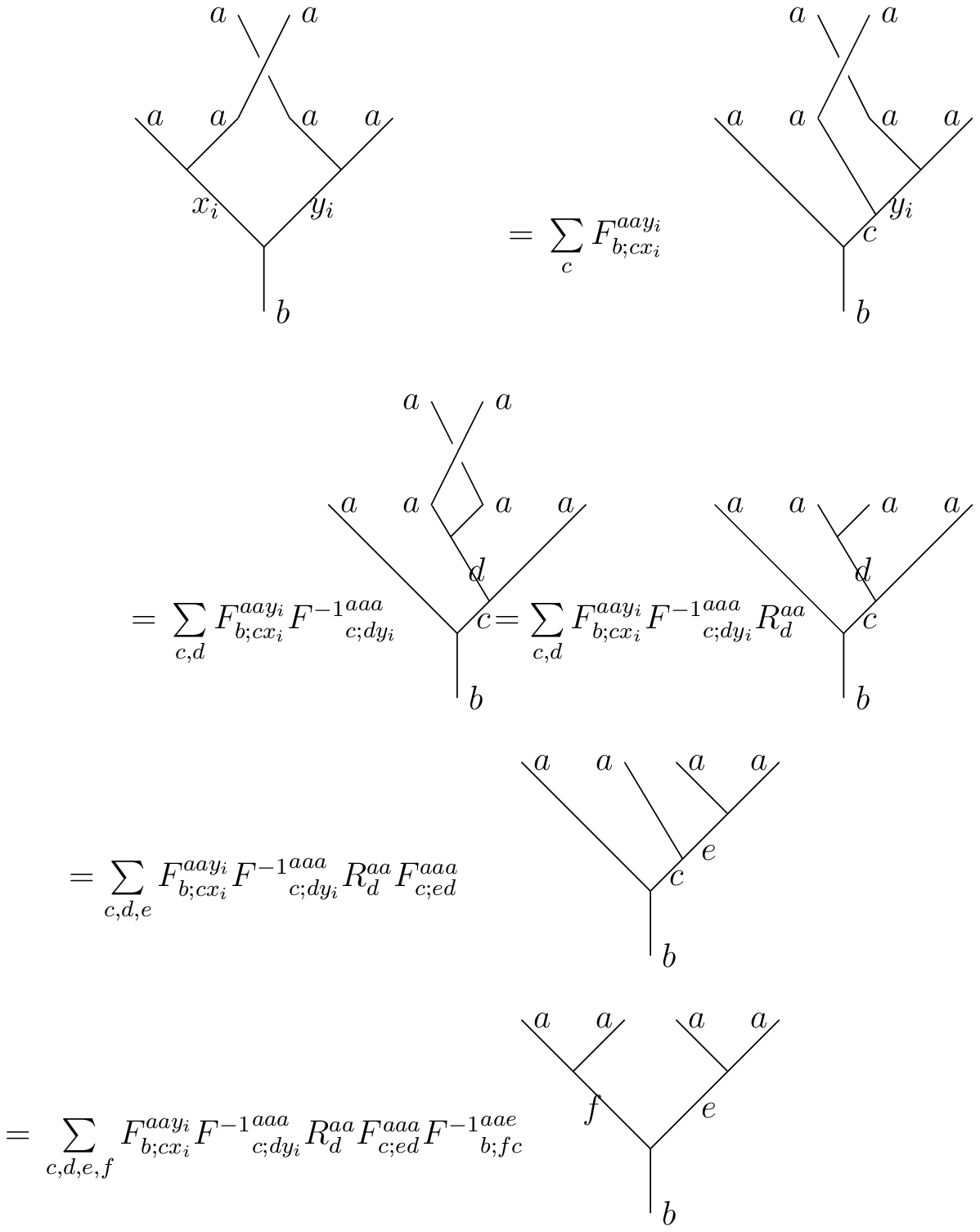}
\caption{Calculating $\sigma_2$} \label{braiding sigma2}
\end{figure}

\bibliographystyle{plain}
\bibliography{UQCMAbib}

\begin{thebibliography}{10}
\bibitem{footnote2}
More powerful gates can be obtained by using higher genus surfaces and exotic
  boundary conditions. we will not use such protocols in this paper.
\bibitem{footnote1}
Potentially, we could miss many intermediate complexity classes.
\end{thebibliography}
\end{document}